\documentclass[12pt]{article}

\usepackage{mathtools}
\usepackage{amssymb}
\usepackage{mathrsfs}
\usepackage{mathtools}
\usepackage{float}
\usepackage{csquotes}

\usepackage{algorithm}
\usepackage{algorithmicx}
\usepackage{algpseudocode}

\usepackage[english]{babel}
\usepackage[english]{babel}
\usepackage{lineno} 
\usepackage{graphicx,multicol}
\usepackage{epic,eepic,epsfig}
\usepackage{amssymb}
\usepackage{amsmath,amsthm}
\usepackage{color}
\usepackage{tikz}

\newcommand{\ay}[1]{{#1}}
\newcommand{\jba}[1]{{#1}}
\usepackage{float}
\usepackage{csquotes}

\usepackage{algorithm}
\usepackage{algorithmicx}
\usepackage{algpseudocode}

\usepackage{eso-pic} 

\newcommand{\induce}[2]{\mbox{$ #1 \langle #2 \rangle$}}

\newcommand{\dom}{\mbox{$\rightarrow$}}

\newtheorem{theorem}{Theorem}

\newtheorem{corollary}[theorem]{Corollary}
\newtheorem{lemma}[theorem]{Lemma}

\theoremstyle{definition}

\newtheorem{case}{Case}
\newtheorem{subcase}{Case}
\numberwithin{subcase}{case}

\newcommand{\defproblem}[3]{
  \vspace{3mm}
\noindent\fbox{
  \begin{minipage}{.95\textwidth}
  \begin{tabular*}{\textwidth}{@{\extracolsep{\fill}}lr} #1  \\ \end{tabular*}
  {\bf{Input:}} #2  \\
  {\bf{Question:}} #3
  \end{minipage}
}
}

\begin{document}
\bibliographystyle{plain}
\title{On the parameterized complexity of 2-partitions\thanks{Research supported by the Independent Research Fond Denmark under grant number DFF 7014-00037B.}
}
\author{
	J. B. Andersen\thanks{Department of Mathematics and Computer Science, University of Southern Denmark, Odense, Denmark (email: jonan15@student.sdu.dk)}
	\and J. Bang-Jensen\thanks{Department of Mathematics and Computer Science, University of Southern Denmark, Odense, Denmark (email: jbj@imada.sdu.dk)}
	\and A. Yeo\thanks{Department of Mathematics and Computer Science, University of Southern Denmark, Odense, Denmark (email: yeo@imada.sdu.dk)}}
	\maketitle

	\begin{abstract}
          We give an FPT algorithm for deciding whether the vertex set a digraph $D$ can be partitioned into two disjoint sets $V_1,V_2$ such that the digraph $D[V_1]$ induced by $V_1$ has a vertex that can reach all other vertices by directed paths, the digraph $D[V_2]$ has no vertex of in-degree zero and $|V_i|\geq k_i$, where $k_1,k_2$ are part of the input. This settles an open problem from \cite{bangTCS640,bangTCS795}.\\
          \noindent{}{\bf Keywords:} FPT algorithm, out-branching, 2-partition, directed graph, parameterized complexity.
	\end{abstract}

	\section{Introduction}
	A {\bf 2-partition} of a digraph $D=(V,A)$ is a partition $(V_1,V_2)$ of $V$ into disjoint sets. Let $\mathcal{P}_1,\mathcal{P}_2$ be two (di)graph properties. Then a {\bf $(\mathcal{P}_1,\mathcal{P}_2)$-partition} of a digraph $D$ is a 2-partition $(V_1,V_2)$ so that
	$D[V_i]$ has property $\mathcal{P}_i$ for $i=1,2$, where $D[V_i]$ is the subdigraph induced by $V_i$. A $(\mathcal{P}_1,\mathcal{P}_2)$-$[k_1,k_2]$-partition of a digraph is as above, but now we also require that $|V_i|\geq k_i$ for $i=1,2$. For example if $\mathcal{P}$ is the property of being acyclic, then the set of digraphs that allow a $(\mathcal{P},\mathcal{P})$-partition are exactly those digraphs that have dichromatic number at most 2. Recognizing such digraphs is NP-complete \cite{bokalJGT46}.

	Problems concerning the existence of certain 2-partitions of a given input digraph has received a lot of attention in the literature, see e.g. 
	\cite{bangTCS640,bangTCS636,bangTCS795,bensmailJCO30,federArXiv1907,kuehnJCT88,lichiardopolIJC2012,misraJCO24,stiebitzKAM,suzukiIPL33,vanthofTCS410}. In the papers \cite{bangTCS640,bangTCS636} the authors gave, for fixed $k_1,k_2$ (not part of the input), a complete complexity classification, in terms of being NP-complete or in XP, for the 120 $(\mathcal{P}_1,\mathcal{P}_2)-[k_1,k_2]$-partition problems corresponding to properties $\mathcal{P}_1,\mathcal{P}_2$ both being one of the following 15 properties:
		being strongly connected, being connected, minimum out-degree at least 1, minimum in-degree at least 1, minimum in- and out-degree at least 1, minimum degree at least 1, having an out-branching, having an in-branching, being acyclic, being complete, being oriented, being independent, being semicomplete, being a tournament and finally being symmetric. 
	They left open to characterize which of those problems that are polynomial actually admit an FPT algorithm. In \cite{bangTCS795} this was settled for the 23 polynomially solvable problems coming from both $\mathcal{P}_1$ and $\mathcal{P}_2$ being one of the following 8 properties: strongly connected, connected, having an out-branching, having an in-branching, having minimum degree at least 1, having minimum in- and out-degree at least 1, being acyclic, being complete. One problem that was left open in \cite{bangTCS795} was to determine the parameterized complexity of deciding the existence of a 2-partition $(V_1,V_2)$ with $|V_i| \ge k_i$ for $i = 1,2$  where the digraph $D[V_1]$ induced by $V_1$ has an out-branching and the digraph $D[V_2]$ has no vertex of in-degree zero.  We prove in this paper that this  problem is FPT.

	\section{Notation and preliminaries}

	The notation we use is consistent with that of \cite{bang2009}. 
	For a digraph $D=(V,A)$ we say that two vertices $u$ and $v$ are \textbf{adjacent} if at least one of the arcs $uv,vu$ is in $A$.
	For a set of vertices $U\subseteq V$, the \textbf{subdigraph induced by} $U$, denoted $D[U]$, is the digraph obtained from $D$ by deleting all vertices $V\setminus U$ and all arcs adjacent to those vertices. For digraph $D$ we denote by $|D| = |V(D)|$ the number of vertices in the graph. We use the same notation for paths and cycles, so $|P|$ is the number of vertices in the path $P$. Paths and cycles will always be directed. The {\bf girth} of a digraph is the length of a shortest directed cycle in $D$.
	A $(u,v)$-path is a directed path from $u\in V$ to $v\in V$ and the digraph $D$ is \textbf{strongly connected} if it contains a $(u,v)$-path for all ordered pairs of vertices $u,v\in V$.
	A {\bf strong component} of a digraph is a maximal induced subdigraph which is strong. A strong component is {\bf initial} if it has no entering arc in $D$.
	
	 The \textbf{underlying graph}, $U(D)$, of a digraph $D$ is the graph obtained from $D$ by replacing every 2-cycle by one edge and then suppressing all the directions of the other arcs. 
	 A digraph $D$ is \textbf{connected} if $U(D)$ is connected. 
	 For a vertex $u\in V$, we denote by $N(u)$ its \textbf{neighbours}, that is, the set of vertices that are adjacent to $u$. The \textbf{out-degree} 
$d^+_D(u)$ is the number of arcs going out of $u$ in $D$. Similarly, the \textbf{in-degree} $d^-_D(u)$ is the number of arcs going into $u$. We denote by $\delta^-(D)$ the minimum in-degree of a vertex in $D$. For a subset $X\subset V$, $N^+(X)$ denotes the set of \textbf{out-neighbours} of $X$ in $D$. 
	 
	
	An \textbf{out-tree} rooted in $s$ is a connected digraph $T_s^+$ such that $d^-_{T_s^+}(s)=0$ and $d^-_{T_s^+}(u)=1$ for all $u\in V(T_s^+)\setminus \{s\}$. An \textbf{out-branching}, $B^+$, of a digraph $D$ is an out-tree such that $V(T_s^+)=V(D)$.

	A digraph $D$ is \textbf{acyclic} if it contains no induced directed cycles, and it is \textbf{complete} if for every pair of vertices $u,v\in V$ induce a 2-cycle $uvu$. 
	
	A parameterized problem with parameter $k$ is in the complexity class XP if instances of size $n$ can be solved in time $O(f(k)n^{g(k)})$ for some pair of computable functions $f,g$. So if $k$ is fixed the problem can be solved in polynomial time. \ay{A problem} is Fixed Parameter Tractable (FPT) if it can be solved in time $O(f(k)n^c)$ for some constant $c$ and computable function $f$.
	
	\section{The $(B^+, \delta^- \ge 1)$-$[k_1,k_2]$-partition problem}
		The following is the problem whose complexity status we settle in this paper.\\
			
		\defproblem{{\sc $(B^+, \delta^- \ge 1)$-$[k_1,k_2]$-partition}}{A digraph $D=(V,A)$ and natural numbers $k_1,k_2$}{Is there a 2-partition $(V_1,V_2)$ of $V$ such that $D[V_1]$ has an out-branching \ay{and} $\delta^-(D[V_2])\geq 1$, \ay{where} $|V_i|\ge k_i$ for $i=1,2$?}\\

		In \cite{bangTCS640} the problem was shown to be polynomially solvable for every fixed pair of natural numbers $k_1,k_2$, but the algorithm has a running time $O(n^{f(k_1)})$ and hence is not an FPT algorithm.\\
	
		We begin with a few simple observations. The root of the out-branching is the only vertex that can possibly have in-degree 0 in a yes-instance. Therefore, if there are two or more vertices with in-degree 0 in $D$, it must be a no-instance.
			So in the following we assume the input is a digraph with at most one vertex with in-degree 0. For future reference, if there is such a vertex with in-degree 0 in the input $D$, we will refer to it as $r$, and in that case $r$ must be the root of the out-branching.
	
		Throughout the solution to the $(B^+, \delta^- \ge 1)$-$[k_1,k_2]$-partition problem, we will say that we {\bf grow} some subset of the vertices inside the graph. What we mean by grow is iteratively adding a vertex to the set that is an out-neighbour of a vertex in the set. Sometimes we want to limit this process to only growing the set to a certain size. We formalize the process in Algorithm \ref{grow}.
	
%
%
%
%

	\begin{algorithm}
		\caption{}\label{grow}
		\begin{algorithmic}
			\Procedure{grow}{$D = (V,A),\, S,\, k$} 

				\While{$|S| < k$ and $N^+(S) \not = \emptyset$}
					\State $v \gets$ any vertex in $N^+(S)$ 
					\State $S \gets S \cup \{v\}$
					
				\EndWhile 
				\State \Return $S$
			\EndProcedure
		\end{algorithmic}
	\end{algorithm}

	Note that if we don't want to limit the growth we can set $k = |V|$.
	
	\begin{lemma} \label{growpoly}
		Algorithm \ref{grow} runs in polynomial time.
	\end{lemma}
	\begin{proof}
		Calculating the set of out-neighbours of a set can be done in polynomial time. A vertex is added to $S$ in each iteration of the while loop, so if none of the exit conditions are met before the $(n-1)$st iteration, then after that iteration every vertex in the graph has been added to $S$ and thus its neighbourhood must be empty. So the algorithm runs in polynomial time.
	\end{proof}
	
	If $\min\{k_1,k_2\}=0$, then we can solve the {\sc $(B^+, \delta^- \ge 1)$-$[k_1,k_2]$-partition} problem in polynomial time. If $k_2 = 0$ we can try all possible roots of the out-branching (there is only one possible root if $r$ exists), and grow it as large as possible by using Algorithm \ref{grow}, with $S$ containing the single vertex $s$ that we are trying as root and $k = |V|$. As we start with a single vertex and each vertex that is added has at least one arc into it from the previous set, it follows by induction that the subdigraph $\induce{D}{S}$ induced by the set $S$ returned by Algorithm \ref{grow} will contain an out-branching $T^+_s$. If one of the possible roots grows to a set $S$ with  $|S| \ge k_1$, then because none of the vertices in $S$ has arcs to vertices in $V \setminus S$, it follows from our assumption that every vertex of $V$ (except possibly $r$) has in-degree at least 1, that we have $\delta^-(D-S) \ge 1$. Thus, $(V_1, V_2) = (S, V \setminus S)$ is a solution.
	
	For the case $k_1 = 0$, we will use the following algorithm which will prove useful later as well.
	\begin{algorithm}
		\caption{}\label{remove_0_degree}
		\begin{algorithmic}
			\Procedure{trim}{$D = (V,A)$} 
			\State $D' \gets$ copy of $D$ 
			\State $V' \gets V(D')$
			\State $Z \gets \{ \, v \in V' \mid d^-(v) = 0 \, \}$
			\While{$Z \not = \emptyset$}
				\State $D' \gets D' - Z$
				\State $V' \gets V(D')$
				\State $Z \gets \{ \, v \in V' \mid d^-(v) = 0 \, \}$			
			\EndWhile
			\State \Return $D'$
			\EndProcedure
		\end{algorithmic}
	\end{algorithm}
	
%
	Algorithm \ref{remove_0_degree} simply iteratively removes vertices with in-degree 0 in the current digraph and return the resulting digraph.
	Note that the resulting digraph $D'$  will have $\delta^-(D') \ge 1$.
	\begin{lemma}
		Algorithm \ref{remove_0_degree} finds the largest possible subdigraph with $\delta^- \ge 1$ in polynomial time.
	\end{lemma}

	\begin{proof}
%
%
%
%
%
		
		It is clear that the algorithm runs in polynomial time and the other part of the claim follows from the fact that a vertex is only removed if it cannot be part of any subdigraph of minimum in-degree at least 1.
	\end{proof}
	
	To handle the case $k_1 = 0$ we simply apply Algorithm \ref{remove_0_degree} to the input $D$.
	Note that the removed vertices (if any) induce a subdigraph with an out-branching rooted at $r$, because a vertex only has in-degree 0 if it is $r$ or if its in-neighbours were removed, implying that an arc into it existed in $D$. Hence if the output $D'=(V', A')$ of has $|V'| \ge k_2$ then $(V \setminus V', V')$ is a solution.\\

	From now on we assume that $k_1,k_2\geq 1$. Let us call a 2-partition $(V_1,V_2)$ of $V$ {\bf good} if $D[V_1]$ has an out-branching and $\delta^-(D[V_2]) \ge 1$.
	As we have assumed that we have at most one  vertex with in-degree 0, the following lemma, which is the basis of the polynomial algorithm of Theorem 3.9 in \cite{bangTCS640}, shows that in order to verify that $(D,k_1,k_2)$ is a 'yes'-instance we only need to find an induced subdigraph $D'$ of $D$ such that $(D',k_1,k_2)$ is a 'yes'-instance.
	
	\begin{lemma}
		\label{sub_solution}
		
		Let $I=(D=(V,A), k_1, k_2)$ be an instance of the $(B^+, \delta^- \ge 1)$-$[k_1,k_2]$-partition problem with at most 1 vertex with in-degree 0. Then any 2-partition $(V_1', V_2')$ (with $r\in V_1'$ if $r$ exists) of a subset of $V$, where
		\begin{itemize}
			\item $D[V_1']$ has an out-branching
			\item $\delta^-(D[V_2']) \ge 1$
			\item $|V_i'| \ge k_i$, $i=1,2$
		\end{itemize}
		can be extended to a good partition of $V$ in polynomial time.
	\end{lemma}
	\begin{proof}
		Call a 2-partition $(V_1', V_2')$ that satisfies the conditions of the lemma a {\bf subsolution}. 
		
		Fix any subsolution.
		Use Algorithm \ref{grow} on the graph $D-V_2'$, starting with $S=V_1'$ and $k=|D - V_2'|$, and denote the result by $V_{T^+}$. In other words we grow the out-branching as large as possible, while not using any vertex in $V'_2$. We claim that $(V_{T^+}, V \setminus V_{T^+})$ is a solution. Clearly $|V_{T^+}| \ge |V_1'| \ge k_1$.
			Since $D[V_1']$ has an out-branching and Algorithm \ref{grow} only adds a vertex $v$ if it has an in-neighbour in the current set $S$, $V_{T^+}$ must contain an out-branching. Secondly, because $r \in V_1' \subseteq V_{T^+}$ if it exists, we know that all vertices in $V_2 = V \setminus V_{T^+}$ had in-degree at least 1 in $D$. It was also the case that $\delta^-(D[V_2']) \ge 1$. Now because $V_{T^+}$ was grown as large as possible in $D-V_2'$, the vertices in $V_2 \setminus V_2'$ (if any) were not reachable from $V_{T^+}$ and thus must still have in-degree at least 1 in $D[V_2]$. Moreover, $V_2' \subseteq V_2$ meaning $|V_2| \ge |V_2'| \ge k_2$, and thus $(V_{T^+}, V_2)$ is a solution.
		
		As this is a simple application of Algorithm \ref{grow} a subsolution can be extended to a solution in polynomial time.
	\end{proof}

	So from here on, finding a subsolution is sufficient to solve the problem.\\

	The following result, due to Shen, will be  used in the proof of our main result.
	\begin{theorem}\cite{shenDM211}
		\label{girth}
		Suppose $G$ is a digraph of order $n$ and girth $g$ with $\delta^+(G) \ge 1$. Let $t = |\{ \, u \in G \mid d^+(u) = 1 \, \}|$. Then
		
		\[ 
		g \le 
		\begin{cases}
			\lceil n / 2\rceil & \text{if }\, t = 0, \\
			\lceil (n + t - 1) / 2\rceil & \text{if }\, t \ge 1
		\end{cases}
		\]
	\end{theorem}
	
	The result also holds if we let $t$ be the number of vertices with in-degree 1 \ay{(in a digraph with $\delta^-(G) \ge 1$)}, instead of out-degree 1 (just reverse all arcs and apply Theorem \ref{girth}). We are now ready to prove our main result.
	
	
	\begin{theorem} \label{mainthm}
		The $(B^+, \delta^- \geq 1)$-$[k_1,k_2]$-partition problem is FPT.
	\end{theorem}
	\begin{proof} 
		We will describe an FPT algorithm which correctly decides whether the input $(D,k_1,k_2)$ is a yes-instance of the $(B^+, \delta^- \geq 1)$-$[k_1,k_2]$-partition problem.
		
		From the previous observations, we can assume that the input $(D,k_1,k_2)$ satisfies that $\min\{k_1,k_2\} \ge 1$ and $D$ has at most one vertex with in-degree 0.
			Let $k = \max(k_1, k_2)$ and define the functions $f,h$ by setting $f(k) = 32 k^3 + 4k$ and $h(k) = 2k \cdot f(k)$.\\

		We start by determining some arcs that cannot be part of the out-branching in any solution. We say that an arc $uv\in A$ is {\bf non-branchable} if $D-\{u,v\}$ does not have a subdigraph $D''$ of size at least $k_2$ with $\delta^-(D'') \ge 1$. Thus if an arc $uv$ is non-branchable, then at least one of the vertices $u,v$ must belong to $V_2$ in any good partition $(V_1,V_2)$. Using Algorithm \ref{remove_0_degree} we can identify the set $B^\dagger\subset A$ of non-branchable arcs in polynomial time.
Let $B = A \setminus B^\dagger$ be the potential branching arcs and let $D_B=(V,B)$ be the subdigraph of $D$ containing exactly the arcs in $B$.  For each vertex $u \in V$ we now define $N_B^+(u)$ as the out-neighbours of $u$ in $D_B$.
		
\[N_B^+(u) = N^+_{D_B}(u)=\{\, v \mid uv \in B \,\}\]

\noindent{}Clearly, $N_B^+(u)$ can be calculated for each $u\in V$ in polynomial time.
We now distinguish several cases which cover all the possibilities.
		
		\begin{case} 
			$\forall u \in V :|N_B^+(u)| \le h(k)$: \label{case:all_possible_outbranchings}
			
			In this case we can check all possible out-trees of size exactly $k_1$ in $D_B$. We do this by trying each possible root, of which there are at most $n = |V|$ (there is only 1 if $r$ exists), and for each root the out-tree will have height at most $k_1$. 
			
			For each vertex, already included in the out-tree, we can use between 0 and $k_1$ of at most $h(k)$ different arcs leaving that vertex. There are at most
			\[\sum_{i=0}^{k_1} {\binom{h(k)}{i}} \]
			ways to do so.
			
			So a very rough upper bound on the total number of out-trees that must be checked is 
			\[n \cdot \left(\sum_{i=0}^{k_1} {\binom{h(k)}{i}}\right)^{k_1} \]
			which means if each check can be done in polynomial time, we can solve this case in FPT time.
			
			For each possible out-tree $T^+$ with $k_1$ vertices we run Algorithm \ref{remove_0_degree} on $D-V(T^+)$. If the resulting graph $D'$ has at least $k_2$ vertices, then $(V_1', V_2') = (V(T^+), V(D'))$ is a subsolution. As Algorithm \ref{remove_0_degree} is indeed a polynomial time algorithm, we conclude that Case \ref{case:all_possible_outbranchings} can be solved in FPT-time.
		\end{case} 
		
		For the case where $\max_{u \in V} \{|N_B^+(u)|\} > h(k)$, we first split the case into whether or not the vertex $r$ exists, and begin with the case where $r$ does not exist. 
		
		\begin{case} 
			$\delta^-(D) \ge 1$ and $\exists s \in V: |N_B^+(s)| > h(k)$: \label{case:many_neighbours_no_r}
			
			We will show that, if we are in this case then $(D,k_1,k_2)$  is a yes-instance  and, in fact,  any vertex $s$ with $|N_B^+(s)| > h(k)$, can be the root of the out-branching in a solution.
			
			Fix any $s$ with $|N_B^+(s)| > h(k)$. 
			
			
			We \ay{will} say that we {\bf contract} $v$ into $u$, if the arc $uv$ exists. Contracting $v$ into $u$ means removing the vertex $v$ and adding arcs from $u$ to every out-neighbour $w$ of $v$ (if $uw$ is not already an arc).
			 As it turns out, we will only contract vertices into our choice for the root of the out-branching $s$, and unless otherwise stated,  we do not rely on the vertices that are contracted into $s$ for our solution.
			
			Let $\mathscr{C}$ be the set of strong components of $D-s$ and for each strong component $C\in \mathscr{C}$  let $N_{B,C}^+(s)$ denote the set of \jba{out-neighbours of $s$ inside $C$} using arcs in $B$, that is,
			
			\[ N_{B,C}^+(s) = N_B^+(s) \cap V(C) \]

			Note that $\mathscr{C}$ and $N_{B,C}^+(s)$ can easily be computed in polynomial time.
			
			
			
			\begin{subcase} \label{manyinsingle} 
				$\exists C \in \mathscr{C}: |N_{B,C}^+(s)| \ge f(k)$:
				
				Fix $C \in \mathscr{C}$ such that  $|N_{B,C}^+(s)| \ge f(k)$. First, we look for a subsolution $(V'_1,V'_2)$, where $C$ is contained completely in one of the sets $V'_i$. 
				
				Let $D'$ be the output from Algorithm \ref{remove_0_degree} on input $D-C-s$ \ay{(i.e. $D'$ is
a maximum subdigraph in $D-C-s$ with $\delta^-(D') \geq 1$)}. If $|D'| \ge k_2$, then $(\{s\} \cup V(C), V(D'))$ is a subsolution, as $s$ has at least $f(k) \ge k_1$ out-neighbours in $C$ and $\delta^-(D') \ge 1$. 
				Similarly, we can try Algorithm \ref{grow} on $D_B-C$ with $S=\{s\}$ and $k=k_1$. If the returned vertex set $V_1'$ has size $|V_1'| \ge k_1$, then, because $C$ is a strong component with at least $f(k) \ge k_2$ vertices, $(V_1', V(C))$ is a subsolution. Both cases can clearly be handled in  polynomial time so we may assume that neither of these cases occur, implying that we must split $C$ in some way to obtain a (sub)solution.

				Iteratively contract all trivial initial strong components into $s$ (iteratively these are the vertices with in-degree 0 in the current digraph, \ay{$D-s$}.) and call the resulting digraph $D$ for simplicity.
If $C$ is not an initial strong component of $D$ after this process, then there is a non-trivial strong component $C'$ with a path into $C$. Using Algorithm \ref{grow} on $D-s$ with $S=V(C')$ and $k=k_2$, we obtain a set $V_2'$, where $\delta^-(D[V_2']) \ge 1$ because we started with a non-trivial strong component. We also have $|V_2'| \geq k_2$ and $|V'_2\cap V(C)|<k_2$ thus $V_1' = N_{B,C}^+(s) \setminus V_2'$ has size $|V_1'| \ge f(k) - k_2 \ge k_1$, so $(V_1' \cup \{s\}, V_2')$ is a subsolution which can be found in polynomial time. Hence we can assume that $C$ is an initial strong component after the contraction step.\\

				Now, observe that any cycle in $D[V(C)]$, small or large, that avoids at least $k_1$ vertices of $N_{B,C}^+(s)$, gives rise to a subsolution. Clearly, if such a cycle $O$ contains at least $k_2$ vertices, we immediately have a subsolution by taking an out-star consisting of $s$ and $k_1-1$ vertices from $N_{B,C}^+(s)-O$ for $V_1'$ and the vertices of $O$ for $V_2'$. If $O$ contains less than $k_2$ vertices, then we can grow it until it has $k_2$ vertices, which again leaves at least $f(k)-k_2 \ge k_1$ vertices for an out-branching from $s$ (we can just take  an out-star from $s$ with $k_1-1$ leaves). Hence it suffices to show that such a cycle $O$ indeed exists.\\

				We first check whether $D[V(C)-N_{B,C}^+(s)]$ is acyclic. If this is not the case, then by the previous observation we have a subsolution and we are done. The same conclusion holds if there is a cycle in $C$ which contains only one vertex of $N_{B,C}^+(s)$ (and the existence of such a cycle can easily be checked in polynomial time). Hence we may assume that every cycle in $C$ contains at least two vertices of $N_{B,C}^+(s)$.
				
				To summarize the situation, these are some facts about the situation we are in:
				
				\begin{enumerate}
					\item \jba{$C$ is an initial strong component of $D-s$.} \label{cond:C_initial}
\item $s$ has at least $f(k)$ out-neighbours inside $C$. We will denote this set of out-neighbours by $S = N_{B,C}^+(s)$ for brevity.
					\item Every cycle in $C$ contains at least two vertices  of $S$. \label{cond:every_cycle_use_S}\label{atleast2inS}
					
					\item By definition of the arc set $B$, for each vertex $u \in S$ there exists a subdigraph $D'$ of $D-\{s,u\}$ with $\delta^-(D')\ge 1$ and $|V(D')|\geq k_2$. \label{cond:sub_d_exists} 
					\item $D-C-s$ does not have a subdigraph with $\delta^- \ge 1$ and size at least $k_2$, so some of the vertices of the subdigraph $D'$ from fact \ref{cond:sub_d_exists} belong to  $C$. Additionally, we had from fact \ref{cond:C_initial} that $C$ is initial, so some cycle in $C$  must be part of the subdigraph of minimum in-degree at least one that we are looking for.

					\item The objective is to show that there exists a cycle in $C$ that avoids at least $k_1$ vertices in $S$
				\end{enumerate}
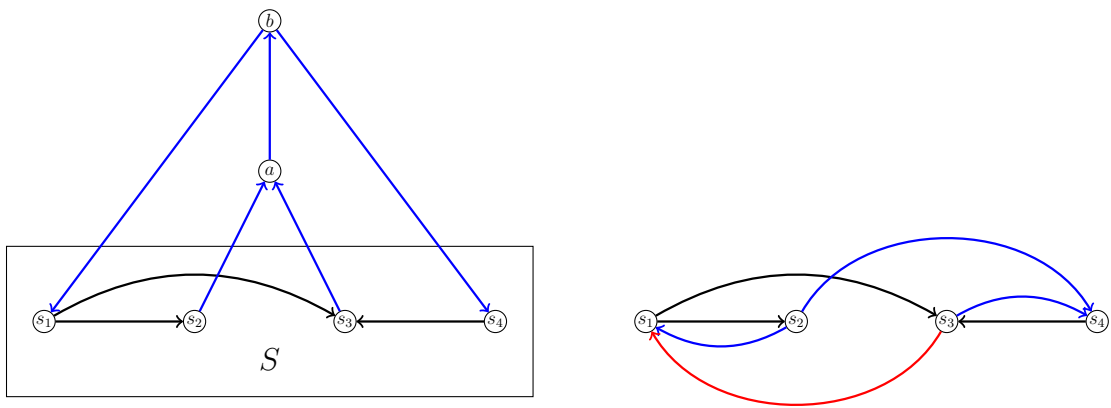
\begin{figure}[H]
\begin{center}
\tikzstyle{vertexB}=[circle,draw, minimum size=14pt, scale=0.6, inner sep=0.5pt]
\tikzstyle{vertexR}=[circle,draw, color=red!100, minimum size=14pt, scale=0.6, inner sep=0.5pt]

\begin{tikzpicture}[scale=1]
\node (s1) at (0,0) [vertexB] {$s_1$};
\node (s2) at (2,0) [vertexB] {$s_2$};
\node (s3) at (4,0) [vertexB] {$s_3$};
\node (s4) at (6,0) [vertexB] {$s_4$};
\draw[->, line width=0.03cm] (s1) to (s2);
\draw[->, line width=0.03cm] (s1) to [out=30,in=150] (s3);
\draw[->, line width=0.03cm] (s4) to (s3);
\draw (-0.5,-1) rectangle (6.5,1);
\node () at (3,-0.5) {$S$};
\node (a) at (3,2) [vertexB] {$a$};
\node (b) at (3,4) [vertexB] {$b$};
\draw[->, line width=0.03cm, color=blue] (s2) to (a);
\draw[->, line width=0.03cm, color=blue] (s3) to (a);
\draw[->, line width=0.03cm, color=blue] (a) to (b);
\draw[->, line width=0.03cm, color=blue] (b) to (s1);
\draw[->, line width=0.03cm, color=blue] (b) to (s4);
\node (s1) at (8,0) [vertexB] {$s_1$};
\node (s2) at (10,0) [vertexB] {$s_2$};
\node (s3) at (12,0) [vertexB] {$s_3$};
\node (s4) at (14,0) [vertexB] {$s_4$};
\draw[->, line width=0.03cm] (s1) to (s2);
\draw[->, line width=0.03cm] (s1) to [out=30,in=150] (s3);
\draw[->, line width=0.03cm] (s4) to (s3);
\draw[->, line width=0.03cm, color=blue] (s2) to [out=210,in=-30] (s1);
\draw[->, line width=0.03cm, color=blue] (s3) to [out=30,in=150] (s4);
\draw[->, line width=0.03cm, color=red] (s3) to [out=240,in=-60] (s1);
\draw[->, line width=0.03cm, color=blue] (s2) to [out=60,in=120] (s4);
\end{tikzpicture}
\end{center}
\caption{An example of the construction of $D_S$. Left: the strong component $C$ with $S=N^+_{B,C}(s)$ and $\induce{D}{S}$ shown inside the rectangle. Right: the final digraph $D_S$ with the \jba{red} arc representing the path $s_3\dom a\dom b\dom s_1$. The other new arcs are shown in blue.}
						\label{fig:D_S}
				\end{figure}

				Let the digraph $D_S$ be obtained as follows, starting from a copy of $D[S]$: For each ordered pair $a,b \in S$, such that there is a path $a \to v_1 \to v_2 \to \ldots \to v_l \to b$ in $C$ where all $v_i \notin S$ for $i\in [l]$, we add the arc $ab$ to $D_S$, if it does not already exist. Note that, $D_S$ does not have parallel arcs. We use $d_S^+(u)$ ($d_S^-(u)$) to denote the out-degree (in-degree) of vertex $u$ in $D_S$.
				
				Let $a,b\in S$ and note that since $C$ is strong it contains an $(a,b)$-path $P$. Let $\langle u_1, \ldots, u_m\rangle$ be those vertices on
 $P$ that are in $S$, listed in the order in which they are visited by $P$. By the construction of $D_S$ it contains the path $a \to u_1 \to \ldots \to u_m \to b$ in $D_S$. Since $a,b$ were arbitrary vertices, it follows that $D_S$ is strong.
				With a similar argument and the fact that every cycle in $C$ contains at least two vertices of $S$, we can conclude that every  cycle in $C$  corresponds to a cycle in $D_S$. 
				We also have that a cycle in $D_S$ corresponds to a closed walk in $C$, so a cycle in $D_S$ that avoids at least $k_1$ vertices in $S$ is verification that there exists a cycle in $C$  which  avoids at least $k_1$ vertices in $S$. 
				
				Let $g=g(D_S)$ denote the girth of $D_S$ and let $t=\min\{t^+,t^-\}$, where $t^+ = |\{ \, u \in D_S \mid d_S^+(u) = 1 \, \}|$, $t^- = |\{ \, u \in D_S \mid d_S^-(u) = 1 \, \}|$. Now we apply Theorem \ref{girth} to $D_S$ (recall that this holds for both $t^+$ and $t^-$). If $t=0$, then because $|S| \ge f(k)>2k_1$ we have $g \le \lceil |S|/2 \rceil \le |S| - k_1$. Otherwise, by Theorem \ref{girth}, we have $g \le \lceil (|S| + t - 1) / 2\rceil$. If $t \le |S|-2k_1$ we make the following calculation.
				
				\[
				\begin{array}{lrcl}
																		& t 	& \le & |S| - 2k_1\\
										\Downarrow 			&			& 		& \\
															& |S| + t 	& \le & 2|S| - 2k_1\\
										\Downarrow 			&			& 		& \\
	& g\leq \lceil (|S| + t - 1) / 2\rceil 	& \le & |S| - k_1,
				\end{array}
				\]
				
				implying that every shortest cycle in $D_s$ avoids at least $k_1$ out-neighbours of $s$ in $D_S$. Clearly we can find a shortest cycle in polynomial time. Hence we may assume that $t > |S|-2k_1$. This implies that at least $|S|-4k_1$ vertices have $d_S^+ = d_S^- = 1$. Denote by $T$ the set of these vertices and let $\bar{T} = S \setminus T$. Consider $D_S[T]$, it will consist of some vertex-disjoint induced paths \ay{(as $\bar{T} \not=\emptyset$, by the definition of $B$)}. 
					Let ${\cal P}$ denote the set of these paths. Clearly ${\cal P}$ can be computed in polynomial time. Note that for each path $P = v_1 \to \ldots \to v_l \in {\cal P}$, there is exactly one pair $a,b \in \bar{T}$ such that $a \to v_1$ and $v_l \to b$ are arcs in $D_S$. Use $a \xrightarrow{P} b$ to denote a path from $a$ to $b$ using the path $P \in {\cal P}$.

				Suppose first that there is a path $P \in {\cal P}$ with at least $k_1$ vertices. We claim that there must be a cycle $W$ in $D_S$ avoiding $P$.
					Assume for the sake of contradiction that there is no such cycle. Then, because every vertex $u$  on $P$  has $d_S^+(u) = d_S^-(u) = 1$, every cycle in $D_S$ must contain the entire path $P$. Thus removing any vertex $v$ of $P$ would destroy all cycles in $D_S$ and hence also in $C$. But because $C$ was initial, this means no vertices in $C-v$ could be part of a subdigraph with $\delta^- \ge 1$, and as we are in a case where $D-C-s$ does not have subdigraph with $\delta^- \ge 1$ and size at least $k_2$, we get a contradiction to the definition of $N_B^+$ (since the arc $sv$ is in $B$). 

Suppose now that every path $ P \in {\cal P}$ has less than $k_1$ vertices. Recall that we have at most $4 k_1$ vertices in $\bar{T}$. We can represent $D_S$ by a directed multigraph $D_{\bar{T}}$ with vertex set $\bar{T}$ and with the arcs $A_{\bar{T}} \cup A_{\cal P}$ where $A_{\bar{T}} = A(D_S[\bar{T}])$ and $A_{\cal P} = \{\, a \xrightarrow{P} b \mid a,b \in \bar{T}, P \in {\cal P}\,\}$. See Figure \ref{Tfig}.

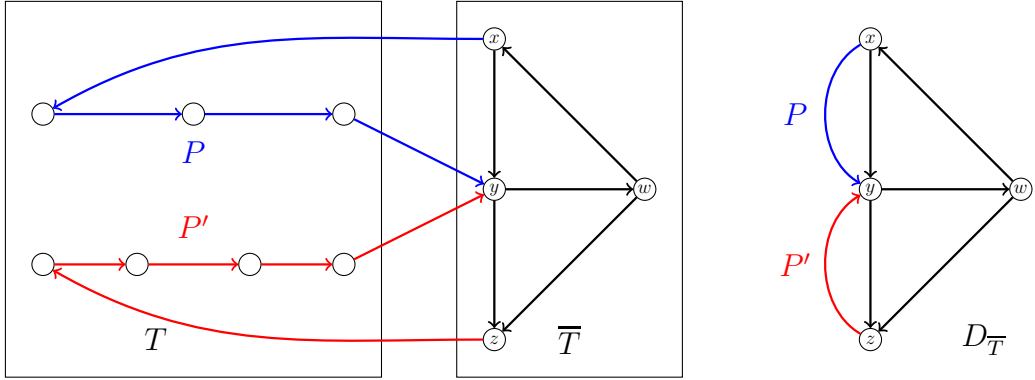
\begin{figure}[H]
\begin{center}
\tikzstyle{vertexB}=[circle,draw, minimum size=14pt, scale=0.6, inner sep=0.5pt]
\tikzstyle{vertexR}=[circle,draw, color=red!100, minimum size=14pt, scale=0.6, inner sep=0.5pt]

\begin{tikzpicture}[scale=1]
\node (x1) at (6,6) [vertexB] {$x$};
\node (y1) at (6,4) [vertexB] {$y$};
\node (z1) at (6,2) [vertexB] {$z$};
\node (w1) at (8,4) [vertexB] {$w$};
\draw [->, line width=0.03cm] (x1) to (y1);
\draw [->, line width=0.03cm] (y1) to (z1);
\draw [->, line width=0.03cm] (w1) to (z1);
\draw [->, line width=0.03cm] (w1) to (x1);
\draw [->, line width=0.03cm] (y1) to (w1);

\node (a1) at (0,5) [vertexB] {};
\node (b1) at (2,5) [vertexB] {};
\node (c1) at (4,5) [vertexB] {};
\node (a2) at (0,3) [vertexB] {};
\node (b2) at (1.25,3) [vertexB] {};
\node (c2) at (2.75,3) [vertexB] {};
\node (d2) at (4,3) [vertexB] {};

\node (x2) at (11,6) [vertexB] {$x$};
\node (y2) at (11,4) [vertexB] {$y$};
\node (z2) at (11,2) [vertexB] {$z$};
\node (w2) at (13,4) [vertexB] {$w$};
\draw [->, line width=0.03cm] (x2) to (y2);
\draw [->, line width=0.03cm] (y2) to (z2);
\draw [->, line width=0.03cm] (w2) to (z2);
\draw [->, line width=0.03cm] (w2) to (x2);
\draw [->, line width=0.03cm] (y2) to (w2);

\draw (5.5,1.5) rectangle (8.5,6.5);
\draw (-0.5,1.5) rectangle (4.5,6.5);
\node () at (1.5,2) {$T$};
\node () at (7,2) {$\overline{T}$};
\draw[->, line width=0.03cm, color=blue] (a1) to (b1);
\draw[->, line width=0.03cm, color=blue] (b1) to (c1);
\draw[->, line width=0.03cm, color=blue] (c1) to (y1);
\draw[->, line width=0.03cm, color=blue] (x1) to [out=180,in=30](a1);
\draw[->, line width=0.03cm, color=red] (a2) to (b2);
\draw[->, line width=0.03cm, color=red] (b2) to (c2);
\draw[->, line width=0.03cm, color=red] (c2) to (d2);
\draw[->, line width=0.03cm, color=red] (d2) to (y1);
\draw[->, line width=0.03cm, color=red] (z1) to [out=180,in=-30] (a2);
\draw[->, line width=0.03cm, color=red] (z2) to [out=150,in=210] (y2);
\draw[->, line width=0.03cm, color=blue] (x2) to [out=210,in=150] (y2);
\node () at (12.5,2) {$D_{\overline{T}}$};
\node () at (2,4.5) [color=blue] {$P$};
\node () at (2,3.5) [color=red] {$P'$};
\node () at (10,5) [color=blue] {$P$};
\node () at (10,3) [color=red] {$P'$};
\end{tikzpicture}

\end{center}
\caption{The left part of the figure  shows an example of the sets $T,\overline{T}$, illustrating that $\induce{D}{T}$ is a collection of vertex disjoint paths. The right part shows the extra (coloured arcs that are added when we create the digraph $D_{\overline{T}}$.}\label{Tfig}
\end{figure}

				For two distinct vertices $a, b$ of $D_{\bar{T}}$, let 
				\begin{equation}
					w_{ab} = \sum_{a \xrightarrow{P} b \in A_{\cal P}} |P|
				\end{equation}
				be the number of vertices in $T$ which lie on paths from $a$ to $b$ in $D_S$. Note that
				\begin{equation} \label{wab}
					\sum_{a,b \in \bar{T}} w_{ab} = |T|
				\end{equation}

				As there are at most $4k_1$ vertices in $D_{\bar{T}}$, there are less than $(4 k_1)^2 = 16 k_1^2 \le 16 k^2$ pairs. We also have that
				\begin{align*}
					|T| &\ge |S|-4k_1\\
					&\ge f(k) - 4k_1\\
					&\ge 32k^3 + 4 k - 4k_1\\
					&\ge 32 k^3					
				\end{align*}
				So it follows from (\ref{wab}) that there must be a pair $a,b$ that has
				\begin{align*}
					w_{ab} &\ge \frac{32 k^3}{16 k^2} \\
					&= 2 k \\
					&\ge 2 k_1
				\end{align*}

				Let ${\cal P}_{ab}\subseteq {\cal P}$ denote the set of those paths that contribute to $w_{ab}$.
				Since each $P \in {\cal P}$ has $|P| < k_1$, we have $|{\cal P}_{ab}|\ge 2$ so if we choose $P'\in {\cal P}_{ab}$ as one of these that uses the fewest vertices, then we always avoid at least $k_1$ vertices in the union of the other paths in ${\cal P}_{ab}$. Because $D_{\bar{T}}$ is strong, there is also a path $P''$ from $b$ to $a$ in $D_{\bar{T}}$ and $P''$ does not intersect any path in ${\cal P}_{ab}$. Hence $P'\cup P''$ is a cycle that avoids at least $k_1$ vertices of $S$.
				
				In conclusion we have shown that if we are in Case (\ref{manyinsingle}), then in polynomial time we can find a cycle which avoids at least $k_1$ vertices of $S$.

			\end{subcase} 
			
			\begin{subcase} 
				$\forall C \in \mathscr{C}: |N_{B,C}^+(s)| < f(k)$:
				
				As there are $h(k)$ out-neighbours of $s$ there must be at least $\frac{h(k)}{f(k)-1} = \frac{2 k \cdot f(k)}{f(k)-1} \ge 2 k$ strong components containing a neighbour of $s$. 
					From the definition of $N_B^+$, we know that $D-s$ has a subdigraph $D'$ with $\delta^-(D') \ge 1$ and at least $k_2$ vertices. We can build such a subdigraph that also leaves at least $k_1$ out-neighbours of $s$ for the out-branching as follows: Starting from an empty set $V_2' = \emptyset$, iteratively add  the vertices of a non-trivial strong component of $D-s-V_2'$ to $V_2'$ and use Algorithm \ref{grow} on $D-s$ with $S=V_2'$ and $k=k_2$ to grow it. We repeat until $|V_2'| \ge k_2$ (when adding a non-trivial strong component we may exceed $k_2$).
				
				As $s$ has an out-neighbour in at least $2 k$ strong components and at most $k_2 \le k$ strong components are added to $D_2'$ during this construction, there are still at least $2 k - k_2 \ge k \ge k_1$ strong components, containing a neighbour of $s$, that can be used to form an out-tree from $s$ of size at least $k_1$. So we can obtain a subsolution in polynomial time.
				
			\end{subcase} 

		\end{case} 
		
		\begin{case} 
			$\exists u \in V: |N_B^+(u)| > h(k)$ and $\exists! r \in V: d^-(r) = 0$:
			
			As we saw in Case \ref{case:many_neighbours_no_r}, if a vertex $s$ has $|N_B^+(s)| > h(k)$ and every other vertex has in-degree at least 1, then we have a yes-instance. 
				So the idea in this case is to start with $r$ and while $|N_B^+(r)| \le h(k)$ contract a neighbour $u \in N_B^+(r)$ into $r$ and recompute $B$ and $N_B^+(r)$. 
			If we reach $N_B^+(r) = \emptyset$ before contracting $k_1-1$ times, then we backtrack and try contracting another neighbour in $N_B^+(r)$, until we have tried all possible out-trees with $k_1$ vertices from $r$, similar to case \ref{case:all_possible_outbranchings}. If we instead reach $|N_B^+(r)| > h(k)$ then the problem is reduced to case \ref{case:many_neighbours_no_r}, and we have a yes-instance.
			
			Because we reduce to Case \ref{case:all_possible_outbranchings} and \ref{case:many_neighbours_no_r} in polynomial time, this case is also solvable in FPT time.
			
		\end{case} 
		
		So we have shown how to find the correct answer in all cases. We also argued that it was possible in FPT time in every case. This concludes the proof.
		
	\end{proof}

	From the proof in Case \ref{case:many_neighbours_no_r} we obtain the following.
	\begin{corollary}
		Let $(D=(V,A), k_1, k_2)$ be an instance of the {\sc $(B^+,\delta^- \ge 1)$-$[k_1,k_2]$-partition}-problem, with $\delta^-(D) \ge 1$. Let $B, N_B^+$ be defined as in the beginning of the proof of Theorem \ref{mainthm} and let $k = \max\{k_1, k_2\}$. If there is a vertex $s \in V$ with $|N_B^+(s)| \ge 64 k^4 + 8 k^2$, then we can find a good 2-partition in polynomial time.
	\end{corollary}

	The {\sc $(B^-,\delta^+ \ge 1)$-$[k_1,k_2]$-partition}-problem is the analogoue of the {\sc $(B^+,\delta^- \ge 1)$-$[k_1,k_2]$-partition}-problem where we want an in-branching in one set of the partition while the other induces a digraph of minimum out-degree at least 1.
		By considering the digraph that we obtain by reversing all arcs we see that the following holds.

	\begin{corollary}
		The {\sc $(B^-,\delta^+ \ge 1)$-$[k_1,k_2]$-partition}-problem is FPT.
	\end{corollary}

	\section{Remarks and open problems}

	If we relax the condition of having an out-branching to that of just being connected, we obtain the following problem.

	\defproblem{{\sc $(connected, \delta^- \ge 1)$-$[k_1,k_2]$-partition}}{A digraph $D=(V,A)$ and natural numbers $k_1,k_2$}{Is there a 2-partition $(V_1,V_2)$ of $V$ such that $D[V_1]$ is connected, $\delta^-(D[V_2])\geq 1$ and $|V_i| \ge k_i$ for $i = 1,2$?}\\

	\begin{theorem} \cite{bangTCS640} \label{conIn1}
		The {\sc $(connected, \delta^- \ge 1)$-$[k_1,k_2]$-partition}-problem \jba{(for fixed $k_1,k_2$)} is NP-complete for general digraphs and polynomially solvable for strong digraphs.
	\end{theorem}

	\begin{theorem}
		The {\sc $(connected, \delta^- \ge 1)$-$[k_1,k_2]$-partition}-problem is FPT for digraphs with minimum in-degree at least 1.
	\end{theorem}

	\begin{proof}
		First observe that if $D=(V,A)$ has minimum indegree at least 1 and $V'_1,V'_2$ are disjoint sets such that 
			$|V'_i|\geq k_i$, $D[V'_1]$ is connected and $\delta^-(D[V'_2])\geq 1$, then we can easily extend this to a 2-partition $(V_1,V_2)$ of $V$ with $V'_i\subseteq V_i$, $i=1,2$ where $D[V_1]$ is connected and $\delta^-(D[V_2])\geq 1$. Hence it suffices to show that we can find a subsolution, if one exists, in FPT time. The proof of this is an easy modification of the proof of Theorem \ref{mainthm}. We leave the details to the interested reader.
	\end{proof}
	
	The theorem also holds for the analogous {\sc $(connected, \delta^+ \ge 1)$-$[k_1,k_2]$-partition}-problem, for digraphs with minimum out-degree at least 1.


	\end{document}